\documentclass{llncs}

\usepackage{
amsmath,
amssymb,
}
\usepackage{latexsym, graphicx}

\newcommand{\Reps}{R^\varepsilon}
\newcommand{\In}{\text{In}}

\newcommand{\Neps}{\mathcal{N}_\varepsilon}

\begin{document}

\title{Bend-Bounded Path Intersection Graphs:\\ Sausages, Noodles, and Waffles on a Grill\thanks{Research support by Czech research grants CE-ITI GA\v{C}R P202/12/6061 and GraDR EUROGIGA GA\v{C}R GIG/11/E023.}}

\author{%
Steven Chaplick\inst{1}\thanks{This author was also supported by the
Natural Sciences and Engineering Research Council of Canada and the University
of Toronto.} \and 
V\'{\i}t Jel\'{\i}nek\inst{2}\and
Jan Kratochv\'{\i}l\inst{3}\and  
Tom\'a\v{s} Vysko\v{c}il\inst{4}
}

\institute{Department of Mathematics, Wilfrid Laurier University, Waterloo, CA,
\email{chaplick@cs.toronto.edu} 
\and
Computer Science Institute, Faculty of Mathematics and Physics,
Charles University, Prague, e-mail: \email{jelinek@iuuk.mff.cuni.cz}
\and
Department of Applied Mathematics, Faculty of
Mathematics and Physics, Charles University, Prague,
\email{honza@kam.mff.cuni.cz}
\and
Department of Applied Mathematics, Faculty of Mathematics and Physics,
Charles University, Prague, e-mail: {\tt whisky@kam.mff.cuni.cz}
}

\date{}
\maketitle

\begin{abstract}
In this paper we study properties of intersection graphs of $k$-bend paths
in the rectangular grid. A $k$-bend path is a path with at most $k$ 90 degree turns.
The class of graphs representable by intersections of $k$-bend paths is denoted by $B_k$-VPG.
We show here that for every fixed $k$, $B_{k}$-VPG $\subsetneq$ $B_{k+1}$-VPG and that recognition of graphs from $B_k$-VPG is NP-complete even
when the input graph is given by a $B_{k+1}$-VPG representation. We also show
that the class $B_k$-VPG (for $k \ge 1$) is in no inclusion relation with the class
of intersection graphs of straight line segments in the plane.
\end{abstract}

\section{Introduction}

In this paper we continue the study of Vertex-intersection graphs of Paths in
Grids\footnote{The grids to which we refer are always rectangular.} (VPG graphs)
started by Asinowski et. al \cite{Asinowski2011,Asinowski2012}. A \emph{VPG
representation} of a graph $G$ is a collection of paths of the rectangular grid
where the paths represent the vertices of $G$ in such a way that two vertices of
$G$ are adjacent if and only if the corresponding paths share at least one
vertex.

VPG representations arise naturally when studying circuit layout problems and
layout optimization \cite{sinden} where layouts are modelled as paths (wires) on
grids. One approach to minimize the cost or difficulty of production involves
minimizing the number of times the wires bend \cite{circuits1,circuits2}. Thus
the research has been focused on VPG representations parameterized by the number
of times each path is allowed to \emph{bend} (these representations are also
the focus of \cite{Asinowski2011,Asinowski2012}). In particular, a
\emph{$k$-bend path} is a path in the grid which contains at most $k$
\emph{bends} where a \emph{bend} is when two consecutive edges on the path have
different horizontal/vertical orientation. In this sense a $B_k$-VPG
representation of a graph $G$ is a VPG representation of $G$ where each path is
a $k$-bend path. A graph is $B_k$-VPG if it has a $B_k$-VPG representation.

Several relationships between VPG graphs and traditional graph classes (i.e.,
circle graphs, circular arc graphs, interval graphs, planar graphs, segment
(SEG) graphs, and string (STRING) graphs) were observed in
\cite{Asinowski2011,Asinowski2012}. For example, the equivalence between string
graphs (the intersection graphs of curves in the plane) and VPG graphs is
formally proven in \cite{Asinowski2012}, but it was known as folklore result
\cite{faithful}.
Additionally, the base case of this family of graph classes (namely, $B_0$-VPG)
is a special case of segment graphs (the intersection graphs of line segments in
the plane). Specifically, $B_0$-VPG is more well known as the 2-DIR\footnote{Note: 
a $k$-DIR graph is an intersection graph of straight line segments in the plane 
with at most $k$ distinct directions (slopes).}.
The recognition problem for the VPG = string graph class is known to be NP-Hard
by \cite{honza} and in NP by \cite{schaeffer}. Similarly, it is NP-Complete to
recognize 2-DIR = $B_0$-VPG graphs \cite{plsat}. However, the recognition status
of $B_k$-VPG for every $k > 0$ was given as an open problem from
\cite{Asinowski2012} (all cases were conjectured to be NP-Complete). We confirm
this conjecture by proving a stronger result. Namely, we demonstrate that
deciding whether a $B_{k+1}$-VPG graph is a $B_k$-VPG graph is NP-Complete (for
any fixed $k > 0$) -- see Section \ref{sec:reduction}.

Furthermore, in \cite{Asinowski2011,Asinowski2012} it is shown that $B_0$-VPG
$\subsetneq$ $B_1$-VPG $\subsetneq$ VPG and it was conjectured that $B_k$-VPG
$\subsetneq$ $B_{k+1}$-VPG for every $k>0$. We confirm this conjecture
constructively -- see Section~\ref{sec:sausage}. 

Finally, we consider the relationship between the $B_k$-VPG graph classes and
segment graphs. In particular, we show that SEG and $B_k$-VPG are incomparable
through the following pair of results (the latter of which is somewhat
surprising): (1) There is a $B_1$-VPG graph which is not a SEG graph; (2) For
every $k$, there is a 3-DIR
 graph which has no $B_k$-VPG representation.

The paper is organized as follows. In Section~\ref{sec:nfl} we introduce the
Noodle-Forcing Lemma, which is the key to restricting the topological structure
of VPG representations\footnote{This was inspired by the order forcing lemma of
\cite{segments}.}. In Section~\ref{sec:sausage} we introduce the ``sausage''
structure which is the crucial gadget that we use for the hardness reduction and
which by itself shows that $B_k$-VPG is strict subset of
$B_{k+1}$-VPG\footnote{This gadget is named due to its VPG representation
resembling sausage links.}. We also demonstrate the incomparability of
$B_k$-VPG and SEG in Section \ref{sec:sausage}. The NP-hardness reduction is
presented in
Section~\ref{sec:reduction}. We end the paper with some remarks and open
problems.

%
%

\section{Noodle-Forcing Lemma}
\label{sec:nfl}

In this section, we present the key lemma of this paper (see Lemma
\ref{lem-noodle}). Essentially, we prove that, for ``proper'' representations
$R$ of a graph $G$, there is a graph $G'$ where $G$ is an induced subgraph of
$G'$ and $R$ is ``sub-representation'' of every representation of $G'$ (i.e.,
all representations of $G'$ require the part corresponding to $G$ to have the
``topological structure'' of $R$). We begin this section with several
definitions. 

Let $G=(V,E)$ be a graph. A \emph{representation} of $G$ is a
collection $R=\{R(v),\,v\in V\}$ of piecewise linear curves in the plane, such
that $R(u)\cap R(v)$ is nonempty iff $uv$ is an edge of~$G$.

An \emph{intersection point} of a representation $R$ is a point in the plane
that belongs to (at least) two distinct curves of $R$. Let $\In(R)$ denote the
set of intersection points of~$R$.

A representation is \emph{proper} if
\begin{enumerate}
 \item each $R(v)$ is a simple curve, i.e., it does not intersect itself,
 \item $R$ has only finitely many intersection points (in particular no two
curves may overlap) and finitely many bends, and
 \item each intersection point $p$ belongs to exactly two curves of $R$, and the
two curves cross in $p$ (in particular, the curves may not touch, and an
endpoint of a curve may not belong to another curve).
\end{enumerate}

Let $R$ be a proper representation of $G=(V,E)$, let $R'$ be another (not
necessarily proper) representation of $G$, and let $\phi$ be a mapping from
$\In(R)$ to $\In(R')$. We say that $\phi$ is \emph{order-preserving} if it is
injective and has the property that for every $v\in V$, if $p_1,p_2,\dotsc,p_k$
are all the distinct intersection points on $R(v)$, then
$\phi(p_1),\dotsc,\phi(p_k)$ all belong to $R'(v)$ and they appear on $R'(v)$ in
the same relative order as the points $p_1,\dotsc,p_k$ on $R(v)$. (If $R'(v)$
visits the point $\phi(p_i)$ more than once, we may select one visit of each
$\phi(p_i)$, such that the selected visits occur in the correct order
$\phi(p_1),\dotsc,\phi(p_k)$.)

For a set $P$ of points in the plane, the \emph{$\varepsilon$-neighborhood of
$P$}, denoted by $\Neps(P)$, is the set of points that have distance less than
$\varepsilon$ from $P$.

\begin{lemma}[Noodle-Forcing Lemma]\label{lem-noodle}
Let $G=(V,E)$ be a graph with a proper representation $R=\{R(v),\,v\in V\}$.
Then there exists a graph $G'=(V',E')$ containing $G$ as an induced subgraph,
which has a proper representation $R'=\{R'(v),\, v\in V'\}$ such
that $R(v)=R'(v)$ for every $v\in V$, and $R'(w)$ is a horizontal or vertical
segment for $w\in V'\setminus V$. Moreover, for any $\varepsilon>0$, any (not
necessarily proper) representation of $G'$ can be transformed by a homeomorphism
of the plane and by circular inversion into a representation
$\Reps=\{\Reps(v),\, v\in V'\}$ with these properties:
\begin{enumerate}
\item for every vertex $v\in V$, the curve $\Reps(v)$ is contained in the
$\varepsilon$-neighbor\-hood of $R(v)$, and $R(v)$ is contained in the
$\varepsilon$-neighborhood of $\Reps(v)$.
\item there is an order-preserving mapping $\phi\colon \In(R)\to\In(\Reps)$,
with the additional property that for every $p\in\In(R)$, the point $\phi(p)$
coincides with the point~$p$.
\end{enumerate}
\end{lemma}

\begin{proof}
Suppose we are given a proper representation $R$ of a graph~$G$. We say that a
point in the plane is \emph{a special point} of $R$, if it is an endpoint of a
curve in $R$, a bend of a curve in $R$, or an intersection point of~$R$.

Before we describe the graph $G'$, we first construct an auxiliary graph $H$
which is a subdivision of a 3-connected plane graph whose drawing overlays the
representation $R$ and has the following properties.

\begin{enumerate}
 \item[P1] The edges of $H$ are drawn as vertical and horizontal segments, and
every internal face of $H$ is a rectangle (possibly containing more than four
vertices of $H$ on its boundary). The outer face of $H$ is not intersected by
any curve of~$R$.
 \item[P2] No curve of $R$ passes through a vertex of $H$, and no edge of $H$
passes through a special point of $R$.
\item[P3] Every face of $H$ contains at most one special point of $R$, and no
two faces containing a special point are adjacent.
\item[P4] Every edge of $H$ is intersected at most once by the curves of~$R$.
\item[P5] Every face of $H$ is intersected by at most two curves of $R$, and if
a face $f$ is intersected by two curves of $R$, then the two curves intersect
inside~$f$.
\item[P6] Every curve of $R$ intersects the boundary of a face of $H$ at most
twice.
\end{enumerate}

We construct the plane graph $H$ in several steps. In the first step, we produce
a square grid $H_0$ such that the whole representation $R$ is contained in the
interior of $H_0$, i.e., no part of $R$ intersects the outer face of $H_0$.
We may further assume that the grid $H_0$ is fine enough so that it has
the properties P1--P3 above. See Fig.~\ref{fig:fiter}.

\begin{figure}[ht]
\centering
\includegraphics[scale=0.7]{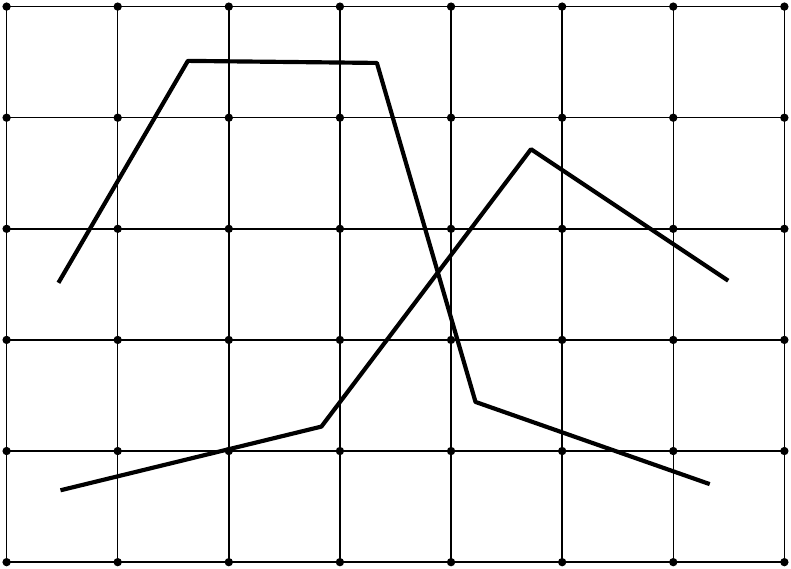}
\caption{An embedding of a graph $H_0$ over a representation $R$.}
\label{fig:fiter}
\end{figure}

If the graph $H_0$ does not satisfy all the properties P1--P6, we will further
subdivide some of the faces of $H_0$. Suppose first that $H_0$ has an edge $e$
that is intersected more than once by the curves of $R$ (see
Fig.~\ref{fig:siter}). We then add a new edge $e'$ into the drawing of $H_0$,
which is parallel to $e$ and embedded very close to~$e$, thus splitting
a face adjacent to $e$ into two new faces. We then split the face incident to
$e$ and $e'$ by new edges perpendicular to $e$ and $e'$, in such a way that
each intersection point of $e$ or $e'$ with a curve of $R$ belongs to a
different edge in the new graph. Performing this operation for every edge of
$H_0$, we obtain a new graph $H_1$, which satisfies the properties P1--P4.

\begin{figure}[ht]
\centering
\includegraphics[scale=0.7]{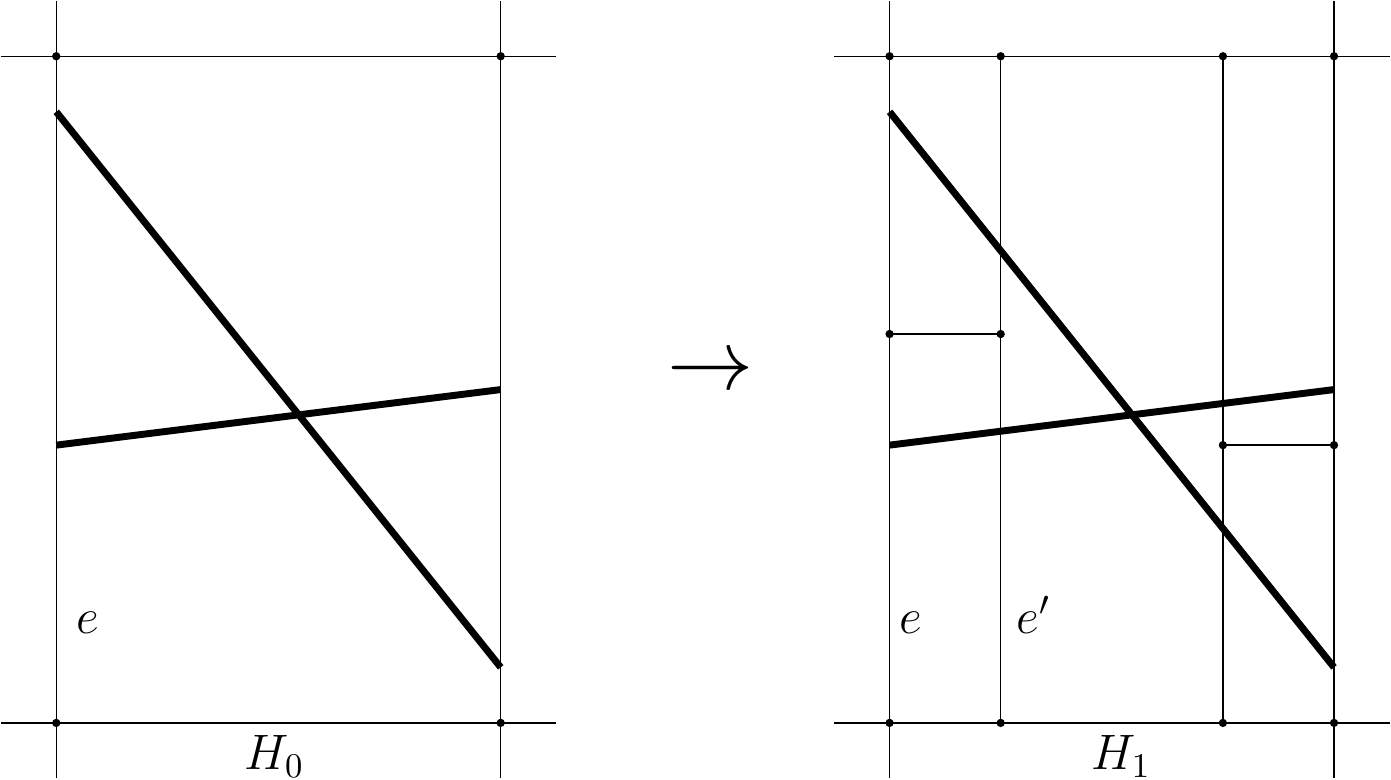}
\caption{An embedding of a graph $H_1$ over a representation $R$.}
\label{fig:siter}
\end{figure}

If the graph $H_1$ fails to satisfy P5 or P6, it means that $H_1$ has a face
$f$ whose intersection with $R$ contains two disjoint curves $p$
and $q$, each of which is a subcurve of a curve from $R$ (see
Fig.~\ref{fig:titer}). In this case, we draw in $f$ a piecewise linear curve $c$
with horizontal and vertical segments which cuts $f$ into two pieces such that
one contains $p$ and the other $q$. We then extend the segments of $c$
to form a grid-like subdivision of the face $f$ into subfaces, each of which is
only intersected by at most one of $p$ and~$q$. If necessary, we may further
subdivide the newly created faces to make sure they do not violate P4. In this
way, we obtain the graph $H$ satisfying all the properties P1--P6.

\begin{figure}[ht]
\centering
\includegraphics[scale=0.7]{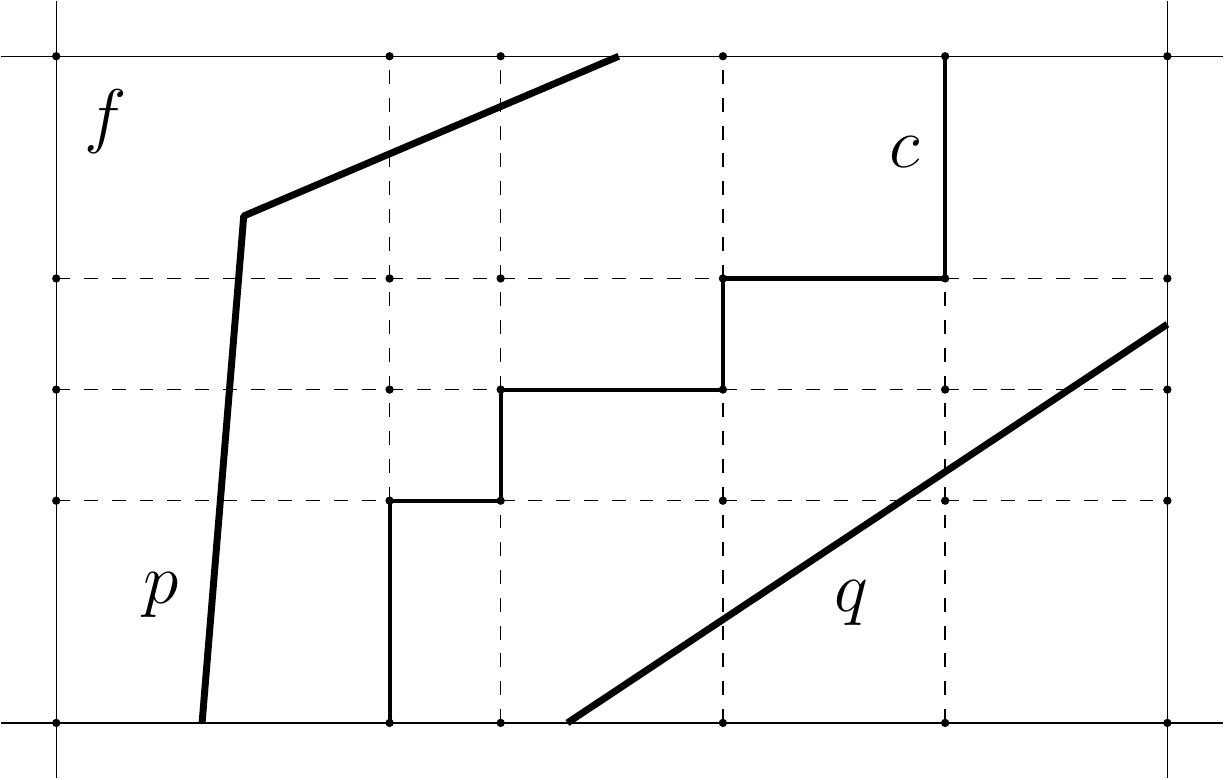}
\caption{A description of the transformation of $H_1$ to $H_2$.}
\label{fig:titer}
\end{figure}

We now transform $H$ into an arrangement $S$ of horizontal and vertical
segments. See Figure~\ref{fig:trans}. For every vertex $v\in V(H)$, the set $S$
contains two segments $S_1(v)$ and $S_2(v)$ which intersect close to the
point~$v$. We call these two segments the \emph{vertex-segments of $v$}. We
assume that the vertex segments do not intersect any of the curves of $R$, and
they do not overlap with any edge of~$H$. For every edge $e\in E(H)$, we further
put into $S$ a segment $S(e)$, called the \emph{edge-segment of $e$}, which
partially overlaps with $e$, intersects exactly those curves of $R$ that $e$
intersects, and does not intersect any vertex-segment or any other edge-segment
of $S$. Finally, for any vertex $v$ that is incident to an edge $e$ of $H$, we
put into $S$ a segment $S(v,e)$ that intersects both $S(e)$ and the vertex
segment of $v$ that is parallel to $S(e)$, and does not intersect any other
segment of $S$ or curve of~$R$. We call $S(v,e)$ \emph{the connector of $v$ and
$e$}.

\begin{figure}[ht]
\centering
\includegraphics[scale=0.7]{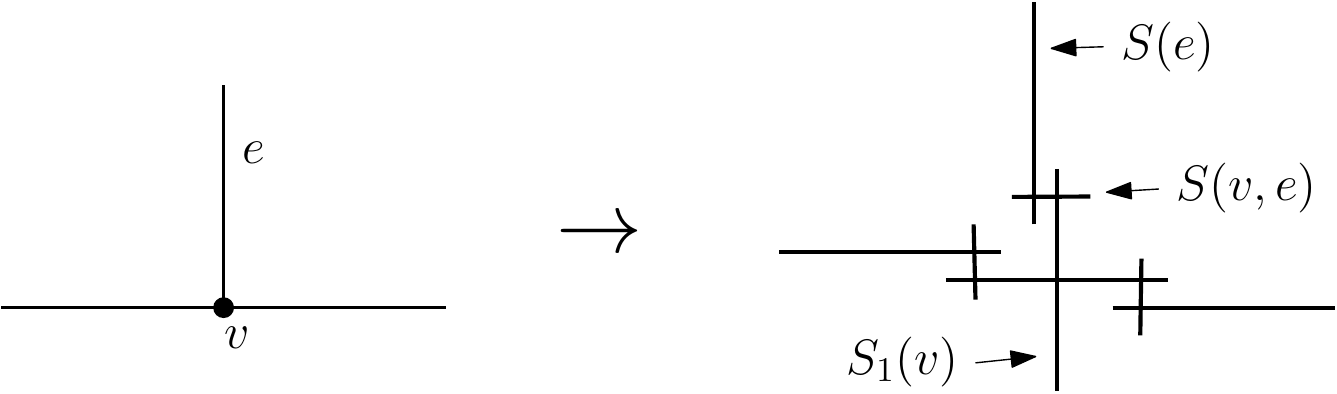}
\caption{Transforming the graph $H$ into a segment arrangement~$S$.}
\label{fig:trans}
\end{figure}

Define a representation $R'=R\cup S$, and let $G'$ be the graph determined
by this representation. Clearly, $G$ is an induced subgraph of $G'$, and the
representation $R'$ has the properties stated in the lemma. It remains to show
that every representation of $G'$ can be transformed into a representation
$\Reps$ that has the two properties stated in the lemma. We will first show
that every representation of $G'$ can be transformed into a representation that
satisfies the first property, and then we will show that the second property
follows from the first one.

Consider an arbitrary representation $R''$ of $G'$. The curves in $R''$ that
correspond to the vertex-segments, edge-segments and connectors will be referred
to as vertex-curves, edge-curves and connector curves. Let $S''(e)$ be an
edge-curve that represents an edge $e\in E(H)$ in $R''$. Let $c_1$ and $c_2$ be
its two adjacent connector curves. Let $\gamma(e)$ be a minimal subcurve of
$S''(e)$ with the property that its two endpoints belong respectively to $c_1$
and $c_2$. In particular, $\gamma(e)$ has no other intersections with the
connector curves apart from its endpoints. Call $\gamma(e)$ \emph{the main part}
of $S''(e)$. Next, for each connector curve, consider its minimal subcurve
whose one endpoint belongs to the main part of its adjacent edge-curve, and
the other endpoint belongs to its adjacent vertex-curve. Call this subcurve
\emph{the main part} of the connector curve.

If, in the representation $R''$, we contract each intersecting pair of
vertex-curves into a single point, and then replace each connector curve and
edge-curve by their main parts, we obtain a drawing of the graph $H$, where a
vertex is represented by the contracted pair of vertex-curves, and an edge
corresponds to the main part of its edge-curve together with the adjacent
connectors. Since $H$ is a subdivision of a 3-connected graph, each of its
drawings can be transformed
into any other drawing by a circular inversion and a homeomorphism. 
In particular,
we may transform $R''$ in such a way that the main part of any edge-curve is
contained inside the corresponding edge-segment $S(e)$, the main part of
any connector curve is contained inside the corresponding connector, and every 
vertex-curve representing a vertex $v\in V(H)$ is embedded in a small
neighborhood of $S_1(v) \cup S_2(v)$. 
Suppose therefore that $R''$ has been transformed in this way.

Let $f$ be an internal face of $H$, whose boundary is formed by $k$ edges
$e_1,\dotsc,e_k$ and $k$ vertices $v_1,\dotsc,v_k$. Consider the union of the
main parts $\gamma(e_1),\dotsc,\gamma(e_k)$, together with the main parts of
their adjacent connectors and together with the vertex curves representing the
vertices $v_1,\dotsc,v_k$. In the union of these curves, there is a unique
bounded region that is adjacent to all of the curves
$\gamma(e_1),\dotsc,\gamma(e_k)$. We call this region the \emph{pseudo-face}
corresponding to~$f$, denoted by $\bar f$. Notice that an edge-curve $R''(e)$
may
only intersect the two pseudo-faces that correspond to the two faces adjacent to
$e$ in~$H$.

Let $v$ be a vertex of $G$. Let $f_0,f_1,\dotsc,f_k$ be the faces of $H$
intersected by the curve $R(v)$ in the order from one endpoint to the other. Let
$e_1,\dotsc,e_k$ be the edges of $H$ crossed by $R(v)$, where $e_i$ is adjacent
to $f_{i-1}$ and~$f_i$. Notice that we have $k\ge 2$ due to the property P3.

Since the curve $R''(v)$ must intersect the edge-curve $S''(e_i)$ for any
$i=1,\dotsc,k$, it must intersect at least one of the two pseudo-faces
$\bar{f}_{i-1}$ and $\bar{f}_i$. Also, since $R''(v)$ does not intersect any
edge-curve other than the $k$ curves $S''(e_i)$ for $i=1,\dotsc,k$, it cannot
enter into any pseudo-face other than $\bar f_i$ for $i=0,\dotsc,k$.

Note, however, that $R''(v)$ does not necessarily intersect the main part of
$S''(e_1)$  and that it does not necessarily penetrate into $\bar f_0$, and
similarly for $S''(e_k)$ and~$\bar f_k$ (see Fig.~\ref{fig:troublecase}).
Overall, we see that $R''(v)$ enters a pseudo-face $\bar f$ if and only if
$R(v)$ enters the face $f$, except possibly for the two faces $f_0$ and $f_k$.
We also see that $R''(v)$ crosses the main part of an edge-curve $S''(e)$ if
and only if $R(v)$ crosses $e$, except possibly for the two edges $e_1$
and~$e_k$. From this it is easy to see that for any $\varepsilon >0$, we may
deform $R''(v)$ so that it will belong to $\Neps(R(v))$, without changing the
boundary of any pseudo-face. Moreover, by possibly creating a `bulge' in
$\gamma(e_0)$ and $\gamma(e_k)$ as shown in Fig.~\ref{fig:troublecase}, we may
ensure that every point of $R(v)$ has a point of $R''(v)$ in its
$\varepsilon$-neighborhood. This shows that $R''$ is homeomorphic to a
representation $\Reps$ having the first property stated in the lemma.

\begin{figure}[ht]
\centering
\includegraphics[width=\textwidth]{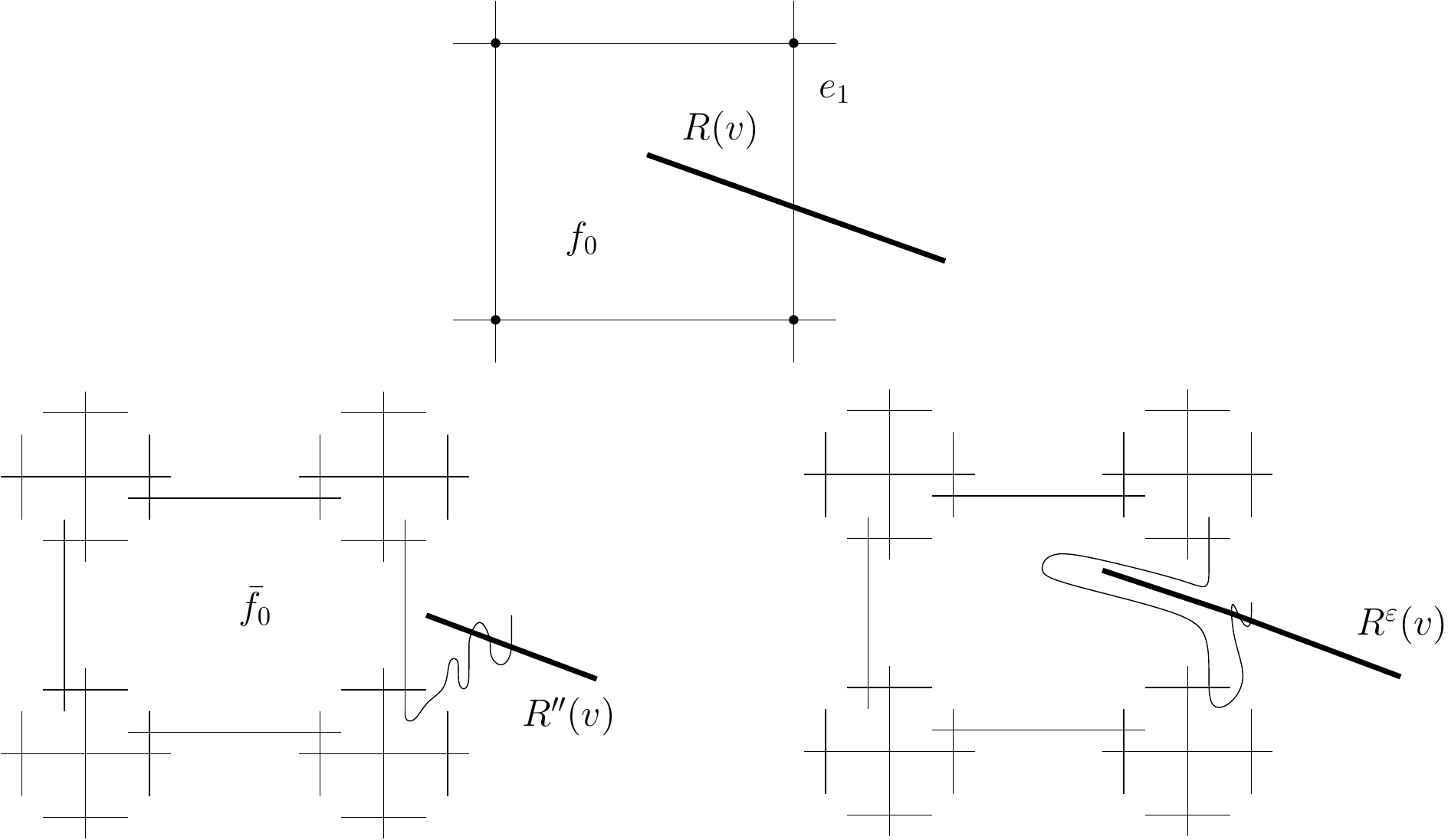}
\caption{A situation where $R''(v)$ does not enter the pseudo-face $\bar f_0$.}
\label{fig:troublecase}
\end{figure}

We will now show that if $\varepsilon$ is small enough, then every
representation $\Reps$ that satisfies the first property in the statement of the
lemma must admit an order-preserving mapping $\phi\colon \In(R)\to\In(\Reps)$.

Let $N(u)$ denote the set $\Neps(R(u))$. We call $N(u)$ \emph{the noodle}
of~$u$. Suppose that two curves $R(u)$ and $R(v)$ have $k$ mutual crossing
points $p_1,\dotsc,p_k$. If $\varepsilon$ is small enough, the intersection
$N(u)\cap N(v)$ has $k$ connected components $P_1,\dotsc,P_k$, each of them
being an open parallelogram, and each $P_i$ containing a unique intersection
point $p_i$. We call $P_i$ \emph{the zone of $p_i$}. By making $\varepsilon$
small, we may assume that the zones of the points in $\In(R)$ are disjoint, that
every noodle has a boundary that is a simple closed curve, and that the boundary
of a noodle does not intersect any of the zones.

Consider now the curve $\Reps(u)$ for some vertex $u\in V$. We may assume that
the endpoints of $\Reps(u)$ coincide with the endpoints of $R(u)$ (otherwise we
may shorten $\Reps(u)$ and deform it in a neighborhood of the endpoints). Choose
an orientation of $R(u)$ (i.e., choose an initial endpoint), and suppose that
$\Reps(u)$ has the same orientation. Suppose that the curve $R(u)$ has $m$
crossing points $q_1,\dotsc,q_m$, encountered in this order. Let
$Q_1,\dotsc,Q_m$ be their zones. Fix a zone $Q_i$ and consider the intersection
$\Reps(u)\cap Q_i$. This intersection is a union of subcurves of $\Reps(u)$.
Choose such a subcurve whose endpoints lie on the opposite sides of $Q_i$; if
there are more such subcurves, choose the first such subcurve visited when
$\Reps(u)$ is traced in the direction of its orientation. Call the chosen
subcurve \emph{the representative of $\Reps(u)$ in $Q_i$}, denoted by $r_i(u)$.
Note that when $\Reps(u)$ is traced from beginning to end, the representatives
are visited in the order $r_1(u),r_2(u),\dotsc,r_m(u)$.

\begin{figure}
\centering
\includegraphics[scale=0.6]{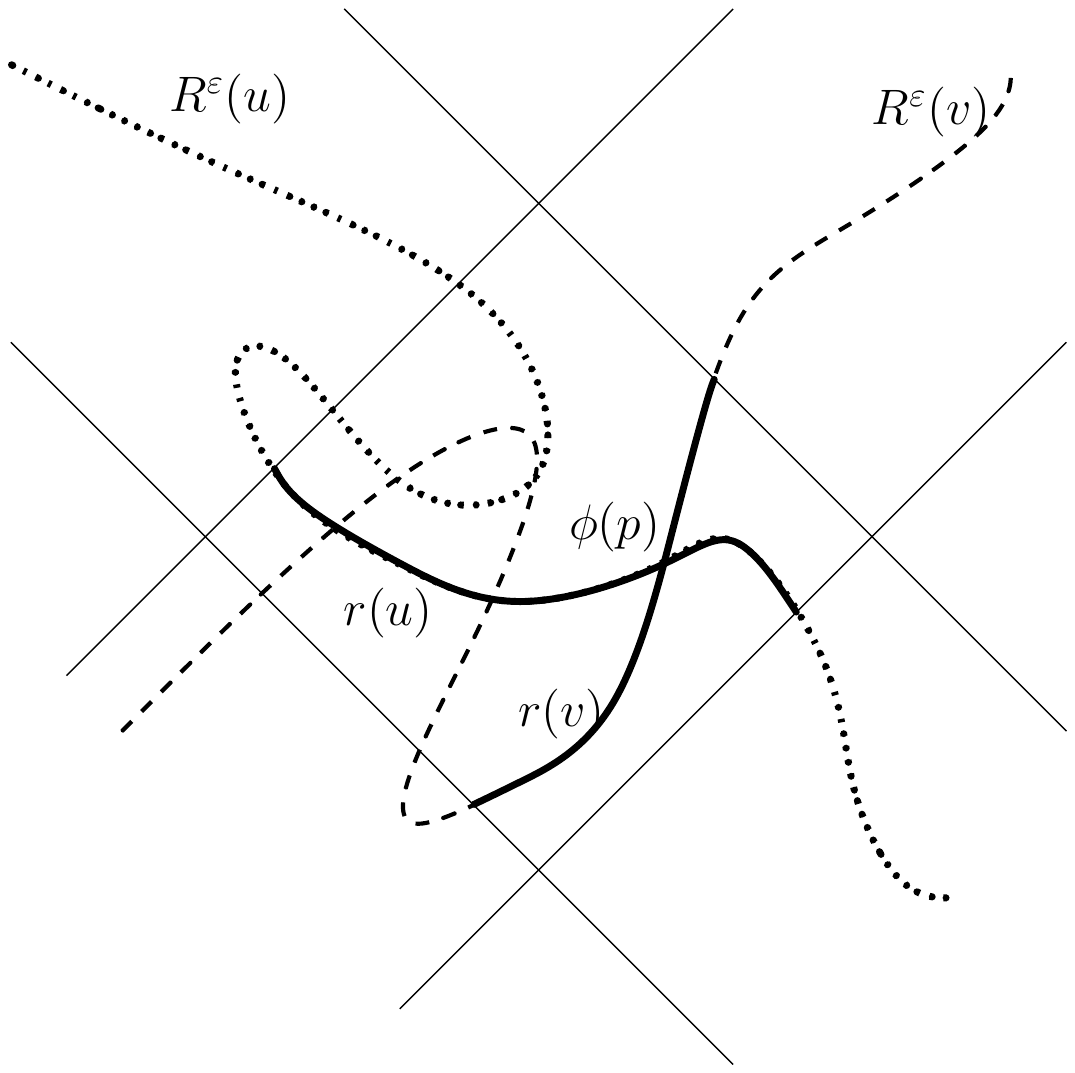}
\caption{Finding a crossing $\phi(p)$ inside the zone of
$p$. The fat parts of the curves are the
representatives.}\label{fig:noodlecross}
\end{figure}

We now define the order-preserving mapping $\phi$. See
Fig.~\ref{fig:noodlecross}. Let $p\in\In(R)$ be an intersection of two curves
$R(u)$ and $R(v)$, and let $P$ be the zone of~$p$. Let $r(u)$ and $r(v)$ be the
representatives of $R(u)$ and $R(v)$ inside~$P$. Define $\phi(p)$ to be an
arbitrary intersection of $r(u)$ and $r(v)$. We may apply a homeomorphism inside
$P$ to make sure that $\phi(p)$ coincides with $p$. By the choice of
representatives, $\phi$ is order-preserving. This proves the lemma.
\qed\end{proof}

%
%

\section{Relations between classes}
\label{sec:sausage}

With the Noodle-Forcing Lemma, we can prove our separation results.

\begin{theorem}\label{thm:grill-K2}
For any $k\ge 1$, there is a graph $G'$ that has a proper representation using
$k$-bend axis-parallel curves, but has no representation using $(k-1)$-bend
axis-parallel curves.
\end{theorem}
\begin{proof}
Consider the graph $K_2$ consisting of a single edge $uv$, with a
representation $R$ in which both $u$ and $v$ are represented by weakly
increasing $k$-bend staircase curves that have $k+1$ common intersections
$p_1,\dotsc,p_{k+1}$, in left-to-right order, see Fig.\ref{fig:sausage}. We refer to this representation as a \emph{sausage} due to it resembling sausage links.

\begin{figure}
\centering
\includegraphics[scale=1]{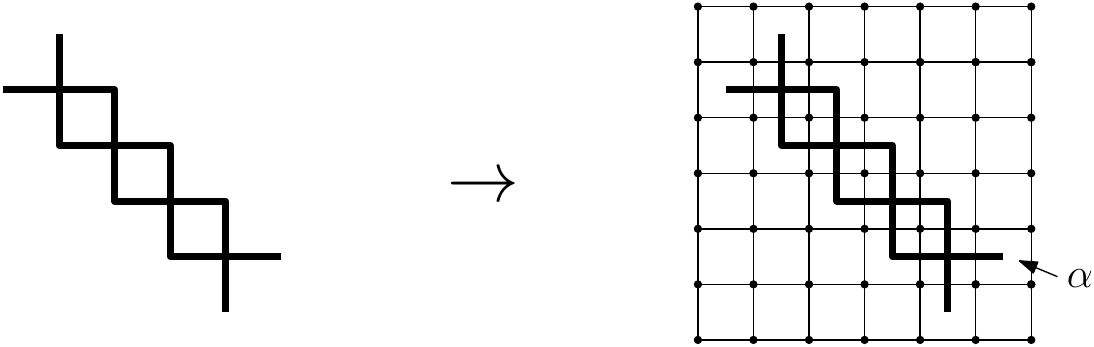}
\caption{The sausage representation for $k = 3$ and its grilled version.}
\label{fig:sausage}
\end{figure}

We now grill the sausage (i.e., we apply the Noodle-Forcing Lemma to $K_2$ and $R$) 
to obtain a graph $G'$ with a $k$-bend representation $R'$.
We claim that $G'$ has no $(k-1)$-bend
representation. Assume for contradiction that there is a $(k-1)$-bend
representation $R''$ of~$G'$. Lemma~\ref{lem-noodle} then implies that there is
an order-preserving mapping $\phi\colon \In(R)\to\In(R'')$. Let $s_i(u)$ be the
subcurve of $R''(u)$ between the points $\phi(p_i)$ and $\phi(p_{i+1})$, and
similarly for $s_i(v)$ and $R''(v)$. Consider, for each $i=1,\dotsc,k$, the
union $c_i=s_i(u)\cup s_i(v)$. We know from Lemma~\ref{lem-noodle} that $s_i(u)$
and $s_i(v)$ cannot completely overlap, and therefore the closed curve $c_i$
must surround at least one nonempty bounded region of the plane. Therefore $c_i$
contains at least two bends different from $\phi(p_i)$ and $\phi(p_{i+1})$. We
conclude that $R''(u)$ and $R''(v)$ together have at least $2k$ bends, a
contradiction.
\qed\end{proof}

A straightforward consequence is the following.
\begin{corollary}\label{cor:sep}
For every $k$, $B_k$-VPG $\subsetneq$ $B_{k+1}$-VPG.
\end{corollary}

Because two straight-line segments in the plane cross at most once, the
Noodle-Forcing Lemma also implies the following.
\begin{corollary}
For every $k \ge 1$, $B_k$-VPG $\not\subset$ SEG.
\end{corollary}

This raises a natural question: Is there some $k$ such that every SEG 
graph is contained in $B_k$-VPG? The following theorem answers it negatively.

\begin{theorem}
\label{thm:3dir}
For every $k$, there is a graph which belongs to 3-DIR but not to $B_k$-VPG.
\end{theorem}

\begin{proof}
We fix an arbitrary $k$. Consider, for an integer $n$, a representation
$R\equiv R(n)$ formed by $3n$ segments, where $n$ of them are horizontal, $n$
are vertical and $n$ have a slope of 45 degrees. Suppose that every two segments
of $R$ with different slopes intersect, and their intersections form the
regular pattern depicted in Figure~\ref{fig:3dir} (with a little bit of creative
fantasy this pattern resembles a waffle, especially when viewed under a linear
transformation).

 \begin{figure}[ht]
 \centering
 \includegraphics[width=\textwidth]{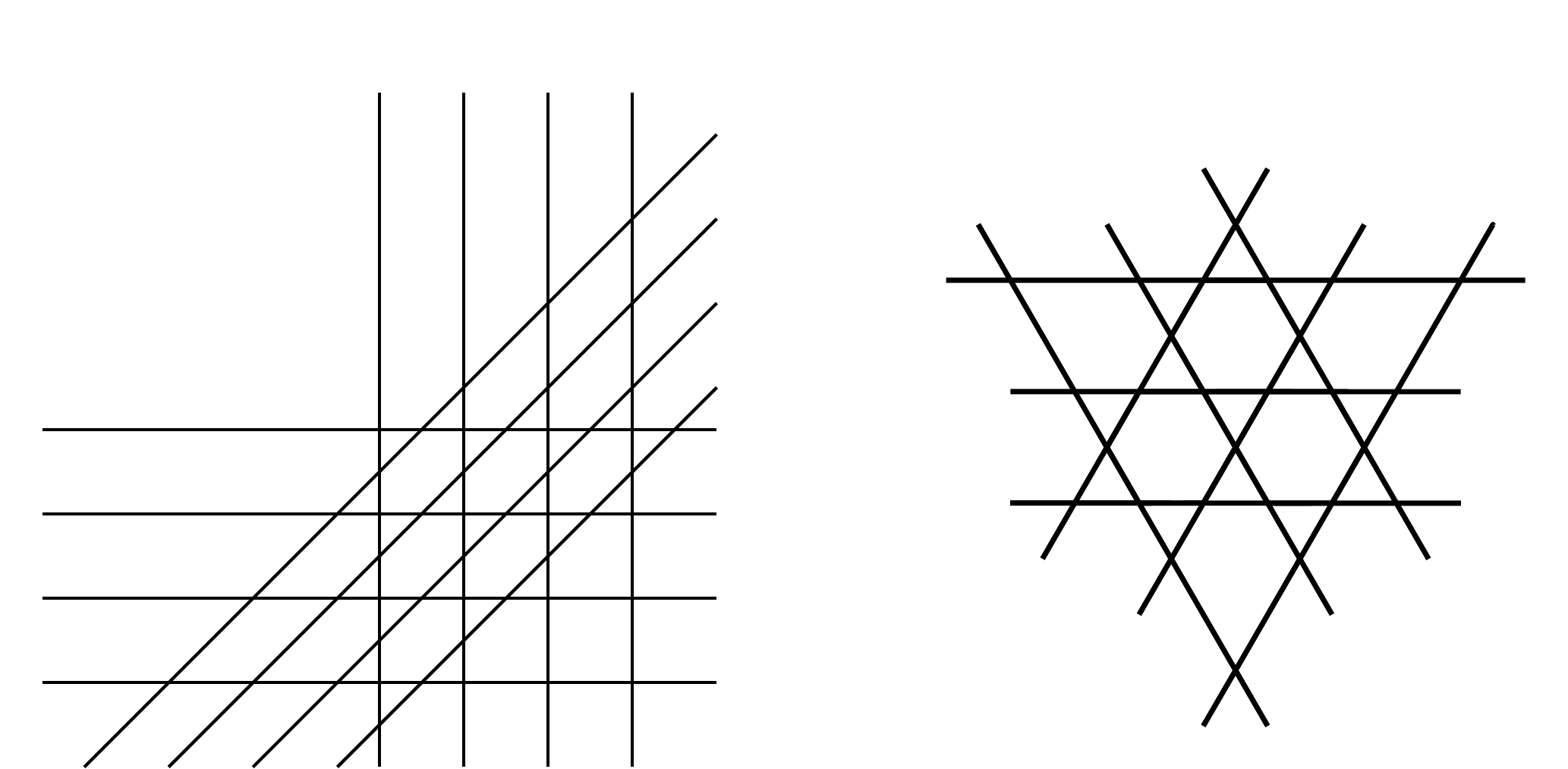}
 \caption{The `waffle' representation $R$ from Theorem~\ref{thm:3dir} and its
 transformed representation.}
 \label{fig:3dir}
 \end{figure}

Note that the representation $R$ forms $\Omega(n^2)$ empty internal
triangular faces bounded by segments of~$R$, and the boundaries of these
faces intersect in at most a single point. Suppose that $n$ is large enough, so
that there are more than $3kn$ such triangular faces. Let $G$ be the graph
represented by~$R$.

The representation $R$ is proper, so we can apply the Noodle-Forcing
Lemma to $R$ and $G$, obtaining a graph $G'$ together with its 3-DIR
representation $R'$. We claim that $G'$ has no $B_k$-VPG representation.

Suppose for contradiction that there is a $B_k$-VPG representation $R''$
of~$G'$. We will show that the $3n$ curves of $R''$ that represent the vertices
of $G$ contain together more than $3kn$ bends.

From the Noodle-Forcing Lemma, we deduce that there exists an
order-preserving mapping~$\phi\colon \In(R_n)\to\In(R''_n)$. Let $T$ be a
triangular face of the representation $R$. The boundary of $T$ consists of three
intersection points $p, q, r\in \In(R)$ and three subcurves $a, b, c$. The three
intersection points $\phi(p)$, $\phi(q)$ and $\phi(r)$ determine the
corresponding subcurves $a''$, $b''$ and $c''$ in~$R''$.

The Noodle-Forcing Lemma implies that there is a homeomorphism $h$ which sends
$a''$, $b''$, and $c''$ into small neighborhoods of $a$, $b$ and $c$,
respectively. This shows that each of the three curves $a''$, $b''$ and $c''$
contains a point that does not belong to any of the other two curves. This in
turn shows that at least one of the three curves is not a segment, i.e., it has
a bend in its interior.

Since the triangular faces of $R$ have non-overlapping boundaries, and since
$\phi$ is order-preserving, we see that for each triangular face of $R$ there
is at least one bend in $R''$ belonging to a curve representing a vertex of~$G$.
Since $G$ has $3n$ vertices and $R$ determines more than $3kn$ triangular faces,
we conclude that at least one curve of $R''$ has more than $k$ bends, a
contradiction.
\qed\end{proof}

%
%

\section{Hardness Results}
\label{sec:reduction}

In this section we strengthen the separation result of Corollary~\ref{cor:sep} by showing
that not only are
the classes $B_k$-VPG and $B_{k+1}$-VPG different, but providing a $B_{k+1}$-VPG
representation does not
help in deciding $B_k$-VPG membership. This also settles the conjecture on
NP-hardness of recognition of
these classes stated in~\cite{Asinowski2012}, in a considerably stronger form than
it was asked.

\begin{theorem}\label{thm:hardness}
For every $k\ge 0$, deciding membership in $B_k$-VPG is NP-complete even if the
input graph is given with a $B_{k+1}$-VPG representation.
\end{theorem}

\begin{proof}
It is not difficult to see that recognition of $B_k$-VPG is in NP and therefore
we will be concerned in showing NP-hardness only. 
We use the NP-hardness reduction developed in~\cite{plsat} for showing that
recognizing grid intersection graphs is NP-complete. Grid intersection graphs
are intersection graphs of vertical and horizontal segments in the plane with
additional restriction that no two segments of the same direction share a common
point. Thus these graphs are formally close but not equal to $B_0$-VPG graphs
(where paths of the same direction are allowed to overlap). However, bipartite
$B_0$-VPG graphs are exactly grid intersection graphs. This follows from a
result of Bellantoni et al.~\cite{whitesides} who proved that bipartite
intersection graphs of axes parallel rectangles are exactly grid intersection
graphs.

\begin{figure}
\centering
\includegraphics[scale=0.8]{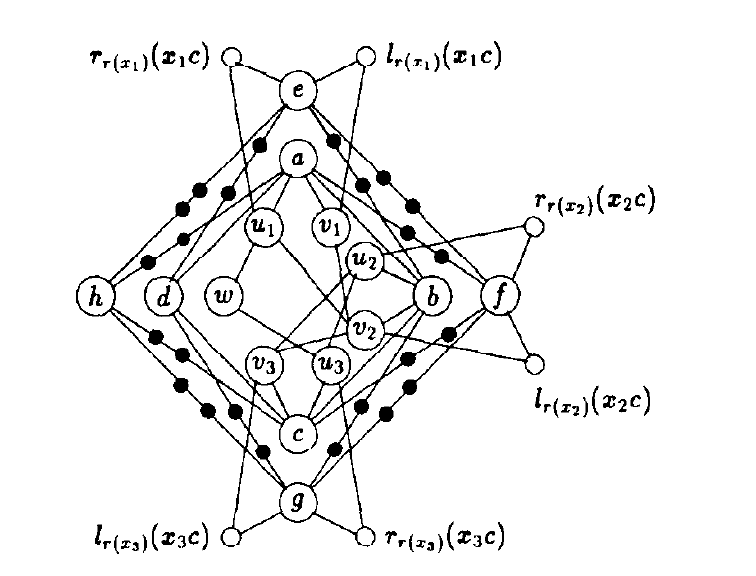}
\caption{The clause gadget reprinted from \cite{plsat}}\label{fig:clausegadget}
\end{figure}

\begin{figure}
\centering
\includegraphics[scale=0.8]{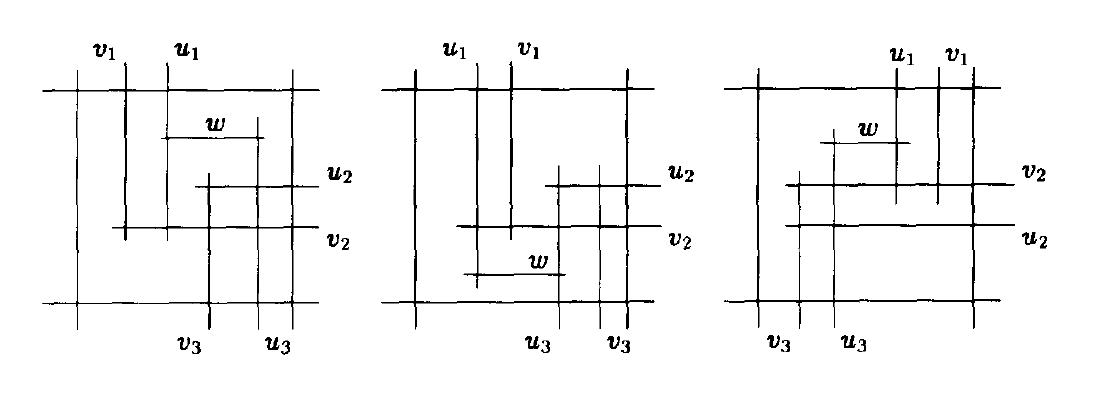}
\caption{The representations of satisfied clauses reprinted from
\cite{plsat}.}\label{fig:onetrue}
\end{figure}

\begin{figure}
\centering
\includegraphics[scale=0.8]{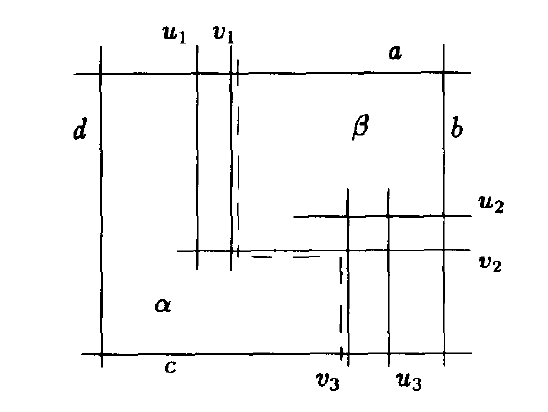}
\caption{The problem preventing the representation of an unsatisfied clause
reprinted from \cite{plsat}.}\label{fig:falseclause}
\end{figure}

\begin{figure}
\centering
\includegraphics[scale=0.7]{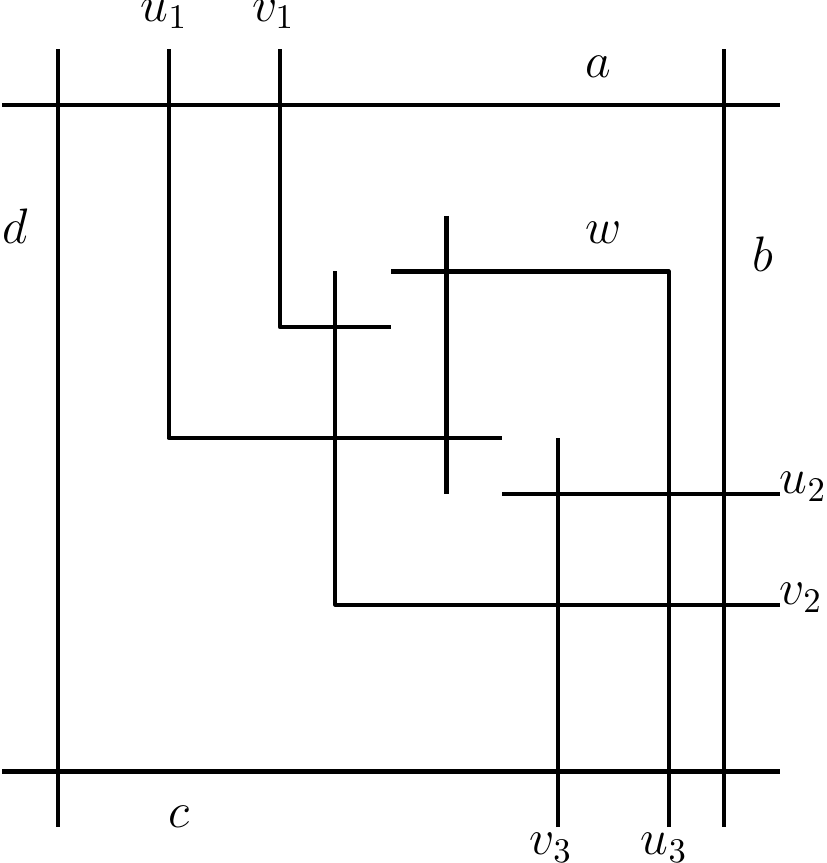}
\caption{The representation of an unsatisfied clause gadget via curves with one
bend.}\label{fig:bend-clause}
\end{figure}

The reduction in~\cite{plsat} constructs, given a Boolean formula $\Phi$, a
graph $G_{\Phi}$ which is a grid intersection graph if and only if $\Phi$ is
satisfiable. In arguing about this, a representation by vertical and horizontal
segments is described for a general layout of $G_{\Phi}$ for which it is also
shown how to represent its parts corresponding to the clauses of the formula,
referred to as {\em clause gadgets}, if at least one literal is true. The clause
gadget is reprinted with a generous approval of the author in
Fig.~\ref{fig:clausegadget}, while Fig.~\ref{fig:onetrue} shows the grid
intersection representations of satisfied clauses, and Fig.~\ref{fig:falseclause}
shows the problem when all literals are false.  In
Fig.~\ref{fig:bend-clause}, we show that in the case of all false literals,
the clause gadget can be represented by grid paths with at most 1 bend each.
It follows that $G_{\Phi}\in B_1$-VPG and a 1-bend representation can be
constructed in polynomial time. Thus, recognition of $B_0$-VPG is NP-complete
even if the input graph is given with a $B_{1}$-VPG representation.

We use a similar approach for arbitrary $k>0$ with a help of the Noodle-Forcing
Lemma. We grill the same representation $R$ of $K_2$ as in the proof of
Theorem~\ref{thm:grill-K2}. We call the resulting graph $P(u)$ where $u$ is one
of the vertices of the $K_2$, the one whose curve in $R$ is ending in a boundary
cell denoted by $\alpha$ in the schematic Fig.~\ref{fig:sausage}. We call this
graph the {\em pin} since it follows from Lemma~\ref{lem-noodle} that it has a
$B_k$-VPG representation such that the bounding paths of the cell $\alpha$ wrap
around the grill and the last segment of $R(u)$ extends arbitrarily far (see the
schematic Fig.~\ref{fig:sausagepin}). We will refer to this extending segment as
the {\em tip} of the pin. It is crucial to observe that in any $B_k$-VPG
representation $R'$ of $P(u)$ all bends of $R'(u)$ are consumed between the
crossing points with the curve representing the other vertex of $K_2$ and hence
the part of $R'(u)$ that lies in the $\alpha$ cell of $R'$ is necessarily
straight.

\begin{figure}
\centering
\includegraphics[scale=0.7]{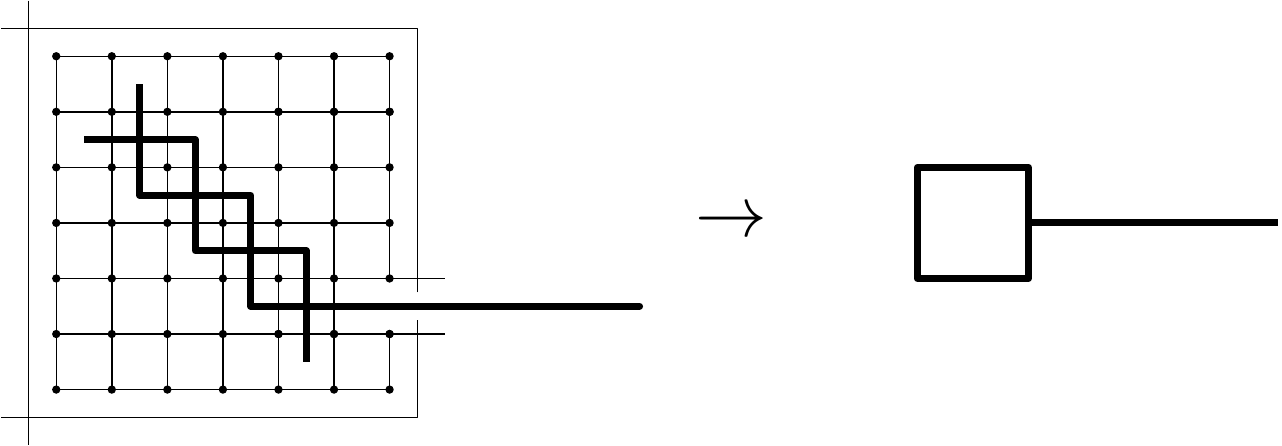}
\caption{Construction of a pin}\label{fig:sausagepin}
\end{figure}

\begin{figure}
\centering
\includegraphics[scale=0.7]{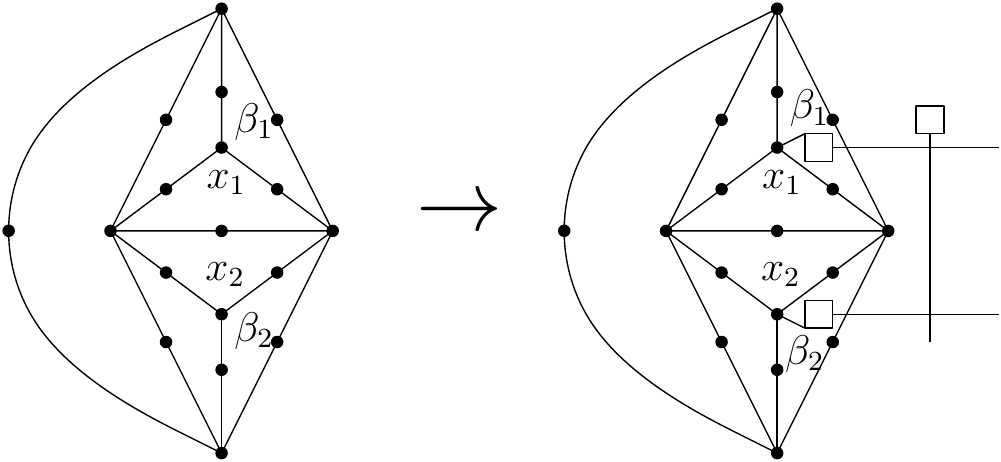}
\caption{Construction of a clothespin}\label{fig:doublepin}
\end{figure}

Next we combine two pins together to form a {\em clothespin}. The  construction
is illustrated in the schematic Fig.~\ref{fig:doublepin}. We start with a $K_4$
whose edges are subdivided by one vertex each. Every STRING representation of
this graph contains 4 basic regions which correspond to the faces of a drawing
of the $K_4$ (this is true for every 3-connected planar graph and it is seen by
contracting the curves corresponding to the degree 2 vertices, the argument
going back to Sinden~\cite{sinden}). We add two vertices $x_1,x_2$ that are
connected by paths of length 2 to the boundary vertices of two triangles, say
$\beta_1$ and $\beta_2$. The curves representing $x_1$ and $x_2$ must lie
entirely inside the corresponding regions. Then we add two pins, say $P(u_1)$
and $P(u_2)$, connect the vertices of the boundary of $\alpha_i$ to $x_i$ by
paths of length 2  and make $u_i$ adjacent to a vertex on the boundary of
$\beta_i$ (for $i=1,2$). Finally, we add a third pin $P(u_3)$ and make $u_3$
adjacent to $u_1$ and $u_2$. We denote the resulting graph by $CP(u)$.

\newcommand{\vR}{{\check R}}

It is easy to check that the clothespin has a $B_k$-VPG representation $\vR$
such that the tips of $\vR(u_1)$ and $\vR(u_2)$ are parallel and extend
arbitrarily far from the rest of the representation, as indicated in
Fig.~\ref{fig:doublepin}.

On the other hand, in any $B_k$-VPG representation $R'$ of $CP(u)$, if a curve
crosses $R'(u_1)$ and $R'(u_2)$ and no other path of $R'(CP(u))$, then it must
cross the tips of $R'(u_1)$ and $R'(u_2)$. This follows from the fact that for
$i=1,2$, $R'(x_i)$ must lie in $\alpha_i$ (to be able to reach all its bounding
curves), and hence, by circle inversion, all bends of $R'(u_i)$ are trapped
inside $\beta_i$. If a curve crosses both $R'(u_1)$ and $R'(u_2)$, it must cross
them outside $\beta_1\cup \beta_2$, and hence it only may cross their tips.

Now we are ready to describe the construction of $G'_{\Phi}$. We take
$G_{\Phi}$ as constructed in~\cite{plsat} replace every vertex $u$ by a
clothespin $CP(u)$, and whenever $uv\in E(G_{\Phi})$, we add edges $u_iv_j,
i,j=1,2$. Now we claim that $G'_{\Phi}\in B_k$-VPG if and only if $\Phi$ is
satisfiable, while $G'_{\Phi}\in B_{k+1}$-VPG is always true.

On one hand, if $G'_{\Phi}\in B_k$-VPG and $R'$ is a $B_k$-VPG representation
of $G'_{\Phi}$, then the tips of $R'(u_1), u\in V(G_{\Phi})$ form a 2-DIR
representation of $G_{\Phi}$ ($R'(u_1)$ and $R'(v_1)$ may only intersect in
their tips) and $\Phi$ is satisfiable.

On the other hand, if $\Phi$ is satisfiable, we represent $G_{\Phi}$ as a grid
intersection graph and replace every segment of the representation by a
clothespin with slim parallel tips and the body of the pin tiny enough so that
does not intersect anything else in the representation. Similarly, if $\Phi$ is
not satisfiable, we modify a 1-bend representation of $G_{\Phi}$ by replacing
the paths of the representation by clothespins with 1-bend on the tips, thus
obtaining a $B_{k+1}$-VPG representation of $G_{\Phi}$. The representation
consists of a large part inherited from the representation of $G_{\Phi}$ and
tiny parts representing the heads of the pins, but these can be made all of the
same constant size and thus providing only a constant ratio refinement of the
representation of~$G_{\Phi}$. The representation is thus still of linear size
and can be constructed in polynomial time.
\qed\end{proof}

\section{Concluding Remarks}

In this paper we have affirmatively settled two main conjectures of Asinowski
et al \cite{Asinowski2012} regarding VPG graphs. We have also demonstrated the
relationship between $B_k$-VPG graphs and segment graphs.

The first conjecture that we settled claimed that $B_k$-VPG is a strict subset
of $B_{k+1}$-VPG for all $k$. We have proven this constructively. Previously
only the following separation was known: $B_0$-VPG $\subsetneq$ $B_1$-VPG
$\subsetneq$ VPG.

The second conjecture claimed that the $B_k$-VPG recognition problem is
NP-Complete for all $k$. We have actually proven a stronger result; namely, that
the $B_k$-VPG recognition problem is NP-Complete for all $k$ even when the input
graph is a $B_{k+1}$-VPG graph. Previously only the NP-Completeness of $B_0$-VPG
(from 2-DIR \cite{plsat}) and VPG (from STRING \cite{honza,schaeffer}) were
known.

Finally due to the close relationship between VPG graphs and segment graphs
(i.e., since $B_0$-VPG = 2-DIR, and SEG $\subsetneq$ STRING = VPG) we have
considered the relationship between these classes. In particular, we have shown
that:
\begin{itemize}
\item There is no $k$ such that 3-DIR is contained in $B_k$-VPG (i.e., SEG
is not contained in $B_k$-VPG for any $k$).
\item $B_1$-VPG is not contained in SEG.
\end{itemize}

Thus, to obtain polynomial time recognition algorithms, one would need
to restrict attention to specific cases with (potentially) useful structure. In
this respect, in \cite{GolRies}, certain subclasses of $B_0$-VPG graphs have
been characterized and shown to admit polynomial time recognition; namely split,
chordal claw-free, and chordal bull-free $B_0$-VPG graphs are discussed in
\cite{GolRies}. Additionally, in \cite{b0-subclasses}, $B_0$-VPG chordal and
2-row $B_0$-VPG\footnote{Where the VPG representation has at most two rows.}
have been shown to have polynomial time recognition algorithms.  In particular,
we conjecture that applying similar restrictions to the $B_k$-VPG graph class
will also yield polynomial time recognition algorithms. It is interesting to
note that since our separating examples are not chordal it is also open whether
$B_k$-VPG chordal $\subsetneq$ $B_{k+1}$-VPG chordal.



%



\begin{thebibliography}{99}

\bibitem{Asinowski2011}
A. Asinowski, E. Cohen, M. C. Golumbic, V. Limouzy, M. Lipshteyn, M. Stern, String graphs of k-bend paths on a grid, Electronic Notes in Discrete Mathematics 37 (2011) pp. 141--146.

\bibitem{Asinowski2012}
A. Asinowski,  E. Cohen,   M. C. Golumbic, V. Limouzy,   M. Lipshteyn, and   M. Stern,
Vertex Intersection Graphs of Paths on a Grid,
Journal of Graph Algorithms and Applications,
16,2  (2012) pp. 129--150.

\bibitem{circuits1}
M. Bandy and M. Sarrafzadeh, Stretching a knock-knee layout for multilayer
wiring, IEEE Trans. Computing, 39 (1990), pp. 148--151.

\bibitem{whitesides}
S. Bellantoni, I. Ben-Arroyo Hartman, T. M. Przytycka, S. Whitesides, Grid intersection graphs and boxicity, Discrete Mathematics 114 (1993) pp. 41--49.

\bibitem{b0-subclasses}
S. Chaplick, E. Cohen, J. Stacho: Recognizing some subclasses of vertex intersection graphs of 0-bend paths in a grid, In: Graph-Theoretic Concepts in Computer Science, Proceedings WG 2011, LNCS 6986, Springer 2011, pp. 319-330.

\bibitem{faithful}
 M. D. Coury, P. Hell, J. Kratochv\'{\i}l, and T. Vysko\v{c}il, Faithful representa-
tions of graphs by islands in the extended grid. In: Proceedings LATIN 2010, LNCS
6034, Springer 2010, pp. 131--142.

\bibitem{bat}
G. Di Battista, P. Eades, R. Tamassia, and I.G. Tollis,
Graph Drawing, Prentice-Hall 1999.

\bibitem{GolRies}
M.C. Golumbic and B. Ries.
On the intersection graphs of orthogonal line segments in the plane: characterizations of some subclasses of chordal graphs.
To appear in Graphs and Combinatorics.

\bibitem{honza}
J. Kratochv\'{\i}l, String graphs II, Recognizing string graphs is
NP-hard, J. Comb. Theory, Ser. B 52 (1991) pp. 67--78.

\bibitem{KM}
J. Kratochv\'{\i}l, J. Matou\v{s}ek, String graphs requiring exponential
representations, J. Comb. Theory, Ser. B 53 (1991) pp. 1--4.

\bibitem{plsat}
J. Kratochv\'{\i}l, A Special Planar Satisfiability Problem and a
Consequence of Its NP-completeness, Discrete Applied Mathematics 52
(1994) pp. 233--252.


\bibitem{segments}
J. Kratochv\'{\i}l and J. Matou\v{s}ek, Intersection Graphs of Segments,
J. Comb. Theory, Ser. B 62 (1994) pp. 289--315.

\bibitem{circuits2}
P. Molitor,  A  survey  on  wiring,  EIK  Journal  of  Information  Processing  and
Cybernetics, 27 (1991) pp. 3--19.

\bibitem{schaeffer}
M. Schaefer, E. Sedgwick, and D. Stefankovic, Recognizing string graphs
in NP, J. Comput. Syst. Sci. 67 (2003) pp. 365-380.

\bibitem{sinden}
F. Sinden, Topology of thin film circuits, Bell System Tech. J., 45 (1966) pp. 1639--1662.

\end{thebibliography}
\end{document}